\documentclass{llncs}

\usepackage{lineno}
\usepackage{amssymb}
\usepackage{amsmath}
\usepackage{multirow}
\usepackage{url}

\newcommand{\ie}{\textit{i.e.}, }

\newcommand*{\wrt}{\textit{w.r.t.}}

\newcommand{\ra}{\mathop{\rightarrow}\limits}
\newcommand{\la}{\mathop{\leftarrow}\limits}
\newcommand{\cons}{\mathop{::}}

\newcommand{\nat}{\mathbb{N}}

\newcommand{\cC}{\mathcal{C}}
\newcommand{\cD}{\mathcal{D}}
\newcommand{\cF}{\mathcal{F}}

\newcommand{\cL}{\mathcal{L}}

\newcommand{\cR}{\mathcal{R}}
\newcommand{\cS}{\mathcal{S}}
\newcommand{\cU}{\mathcal{U}}
\newcommand{\cV}{\mathcal{V}}

\newcommand{\capf}{\textsc{cap}}
\newcommand{\renf}{\textsc{ren}}

\newcommand{\id}{\mathit{id}}

\newcommand{\dg}{\operatorname{\mathit{DG}}}
\newcommand{\scycles}{\operatorname{\mathit{s-cycles}}}

\newcommand{\mgu}{\operatorname{\mathit{mgu}}}
\newcommand{\dpos}{\operatorname{\mathit{DPos}}}
\newcommand{\pos}{\operatorname{\mathit{Pos}}}
\newcommand{\npos}{\operatorname{\mathit{NPos}}}
\newcommand{\mroot}{\operatorname{\mathit{root}}}
\newcommand{\scc}{\operatorname{\mathit{SCC}}}

\newcommand{\gunf}{\operatorname{\mathit{gunf}}}
\newcommand{\gunfop}{\operatorname{\mathit{GU}}}
\newcommand{\select}{\operatorname{\mathit{select}}}
\newcommand{\selecta}{\operatorname{\mathit{select\_all}}}
\newcommand{\selectlm}{\operatorname{\mathit{select\_lm}}}
\newcommand{\selectlmne}{\operatorname{\mathit{select\_lmne}}}
\newcommand{\Var}{\operatorname{\mathit{Var}}}

\newcommand{\termset}{\mathcal{T}(\cF,\cV)}
\newcommand{\termsett}{\mathcal{T}(\cF\cup\cF^{\#},\cV)}

\newcommand{\f}{\mathsf{f}}
\newcommand{\g}{\mathsf{g}}
\newcommand{\h}{\mathsf{h}}
\newcommand{\s}{\mathsf{s}}
\newcommand{\zero}{\mathsf{0}}
\newcommand{\one}{\mathsf{1}}
\newcommand{\two}{\mathsf{2}}

\newcommand{\aprove}{\textsf{AProVE}}
\newcommand{\muterm}{\textsf{MU{-}TERM}}
\newcommand{\natt}{\textsf{NaTT}}
\newcommand{\nti}{\textsf{NTI}}
\newcommand{\wanda}{\textsf{WANDA}}

\title{Guided Unfoldings for Finding
Loops in Standard Term Rewriting}
\author{\'Etienne Payet}
\institute{LIM, Universit\'e de La R\'eunion, France\\
\email{etienne.payet@univ-reunion.fr}}

\begin{document}

\maketitle

\begin{abstract}
  In this paper, we reconsider the unfolding-based technique that we have
  introduced previously for detecting loops in standard term rewriting.
  We modify it by \emph{guiding} the unfolding process, using disagreement
  pairs in rewrite rules. This results in a partial computation of the
  unfoldings, whereas the original technique consists of a thorough
  computation followed by a mechanism for eliminating some rules.
  We have implemented this new approach in our tool \nti{} and conducted
  successful experiments on a large set of term rewrite systems.
\end{abstract}


\section{Introduction}
In~\cite{payet08}, we have introduced a technique for finding
\emph{loops} (a periodic, special form of non-termination)
in standard term rewriting. It consists of unfolding the
term rewrite system (TRS) $\cR$ under analysis and of
performing a semi-unification~\cite{lankford78} test on
the unfolded rules for detecting loops. The unfolding
operator $U_{\cR}$ which is applied processes both forwards
and backwards and considers \emph{every} subterm of the
rules to unfold, including variable subterms.
\begin{example}\label{ex:zero}
  Let $\cR$ be the TRS consisting of the
  following rules ($x$ is a variable):
  \begin{linenomath*}
    \[R_1 = \underbrace{\f(\s(\zero),\s(\one),x)}_{l}
    \ra \underbrace{\f(x,x,x)}_{r} \qquad
    R_2 = \h \ra \zero \qquad
    R_3 = \h \ra \one\;.\]
  \end{linenomath*}
  Note that $\cR$ is a variation of a well-known example by
  Toyama~\cite{toyama87}.
  Unfolding the subterm $\zero$ of $l$ backwards with the rule
  $R_2$, we get the unfolded rule $U_1=\f(\s(\h),\s(\one),x)\ra \f(x,x,x)$.
  Unfolding the subterm $x$ (a variable) of $l$ backwards with $R_2$,
  we get $U_2=\f(\s(\zero),\s(\one),\h)\ra \f(\zero,\zero,\zero)$.
  Unfolding the first (from the left) occurrence of $x$ in $r$ forwards
  with $R_2$, we get $U_3=\f(\s(\zero),\s(\one),\h)\ra \f(\zero,\h,\h)$.
  We have $\{U_1,U_2,U_3\}\subseteq U_{\cR}(\cR)$.
  Now, if we unfold the subterm $\one$ of $U_1$ backwards with $R_3$,
  we get $\f(\s(\h),\s(\h),x)\ra \f(x,x,x)$, which is an element of
  $U_{\cR}(U_{\cR}(\cR))$.
  The left-hand side $l_1$ of this rule semi-unifies with its right-hand side
  $r_1$ \ie $l_1\theta_1\theta_2=r_1\theta_1$ for the substitutions
  $\theta_1=\{x/\s(\h)\}$ and $\theta_2=\{\}$. Therefore,
  $l\theta_1=\f(\s(\h),\s(\h),\s(\h))$ loops with respect to $\cR$
  because it can be rewritten to itself using the rules of $\cR$
  (the redex is underlined at each step):
  \begin{linenomath*}
    \[\f(\s(\underline{\h}),\s(\h),\s(\h))\ra_{R_2}
    \f(\s(\zero),\s(\underline{\h}),\s(\h))\ra_{R_3}
    \underline{\f(\s(\zero),\s(\one),\s(\h))}\ra_{R_1}
    \f(\s(\h),\s(\h),\s(\h))\;.\]
  \end{linenomath*}
\end{example}

Iterative applications of the operator $U_{\cR}$ result in a
combinatorial explosion which significantly limits the
approach. In order to reduce it, a mechanism is introduced
in~\cite{payet08} for eliminating unfolded rules which are
estimated as \emph{useless} for detecting loops. Moreover,
in practice, three analyses are run in parallel (in different
threads): one with forward unfoldings only, one with backward
unfoldings only and one with forward and backward unfoldings
together.

So, the technique of~\cite{payet08} roughly consists in computing
\emph{all} the rules of $U_{\cR}(\cR)$, $U_{\cR}(U_{\cR}(\cR))$,
\dots{} and removing some useless ones, until the semi-unification
test succeeds on an unfolded rule or a time limit is reached.
Therefore, this approach corresponds to a breadth-first
search for a loop, as the successive iterations of $U_{\cR}$
are computed thoroughly, one after the other.
However, it is not always necessary to compute all the elements
of each iteration of $U_{\cR}$. For instance, in Ex.~\ref{ex:zero}
above, $U_2$ and $U_3$ do not lead to an unfolded rule satisfying
the semi-unification criterion. This is detected by the eliminating
mechanism of~\cite{payet08}, but only \emph{after} these two rules
are generated. In order to \emph{avoid} the generation of these
useless rules, one can notice that $l$ and $r$ differ at the first
argument of $\f$: in $l$, the first argument is $\s(\zero)$ while in
$r$ it is $x$. We say that $\langle \s(\zero),x \rangle$ is a
\emph{disagreement pair} of $l$ and $r$.
Hence, one can first concentrate on resolving this disagreement,
unfolding this pair only, and then, once this is resolved, apply
the same process to another disagreement pair.
\begin{example}[Ex.~\ref{ex:zero} continued]
  \label{ex:zero-continued}
  There are two ways to resolve the disagreement pair
  $\langle \s(\zero),x \rangle$ of $l$ and $r$ (\ie make it
  disappear).

  The first way consists in unifying
  $\s(\zero)$ and $x$, \ie in computing $R_1\theta$
  where $\theta$ is the substitution $\{x/\s(\zero)\}$,
  which gives
  $V_0=\f(\s(\zero),\s(\one),\s(\zero))\ra
  \f(\s(\zero),\s(\zero),\s(\zero))$.
  The left-hand side of $V_0$ does not semi-unify with its right-hand
  side.

  The other way is to unfold $\s(\zero)$ or $x$.
  We decide not to unfold variable subterms, hence
  we select $\s(\zero)$. As it occurs in
  the left-hand side of $R_1$, we unfold it backwards.
  The only possibility is to use $R_2$, which results in
  \begin{linenomath*}
    \[V_1=\f(\s(\h),\s(\one),x)\ra \f(x,x,x)\;.\]
  \end{linenomath*}
  Note that this approach only generates two rules
  ($V_0$ and $V_1$) at the first iteration of the unfolding
  operator. In comparison, the approach of~\cite{payet08}
  produces 14~rules (before elimination), as all the subterms
  of $R_1$ are considered for unfolding.

  Hence, the disagreement pair $\langle \s(\zero),x \rangle$
  has been replaced with the disagreement pair
  $\langle \s(\h),x \rangle$ in $V_1$. Unifying $\s(\h)$ and $x$
  \ie computing $V_1\theta'$ where $\theta'$ is the substitution
  $\{x/\s(\h)\}$, we get
  $V'_1=\f(\s(\h),\s(\one),\s(\h))\ra \f(\s(\h),\s(\h),\s(\h))$.
  So, the disagreement $\langle \s(\zero),x \rangle$ is solved:
  it has been replaced with $\langle \s(\h),\s(\h) \rangle$.
  Now, $\langle \one,\h \rangle$ is a disagreement pair in $V'_1$
  (here we mean the second occurrence of $\h$ in the right-hand side
  of $V'_1$). Unfolding $\one$ backwards with $R_3$, we get
  $W=\f(\s(\h),\s(\h),\s(\h))\ra \f(\s(\h),\s(\h),\s(\h))$
  and unfolding $\h$ forwards with $R_3$, we get
  $W'=\f(\s(\h),\s(\one),\s(\h))\ra \f(\s(\h),\s(\one),\s(\h))$.
  The semi-unification test succeeds on both rules: we get the looping
  terms $\f(\s(\h),\s(\h),\s(\h))$ and $\f(\s(\h),\s(\one),\s(\h))$
  from $W$ and $W'$, respectively.
\end{example}

In the approach sketched in Ex.~\ref{ex:zero-continued}, the iterations of
$U_{\cR}$ are not thoroughly computed because only some selected disagreement
pairs are considered for unfolding, unlike in our previous
approach~\cite{payet08} which tries to unfold all the subterms in rules.
Hence, now the unfoldings are \emph{guided} by disagreement pairs.
In this paper, we formally describe the intuitions
presented above (Sect.~\ref{sect:gunf}--\ref{sect:comparison08}).
We also report experiments on a large set of rewrite systems
from the TPBD~\cite{TPDB} (Sect~\ref{sect:experiments}).
The results we get in practice with the new approach are better
than those obtained with the approach of~\cite{payet08} and
we do not need to perform several analyses in parallel, nor to
unfold variable subterms, unlike~\cite{payet08}.

\section{Preliminaries}
If $Y$ is an operator from a set $E$
to itself, then for any $e\in E$ we let
\[(Y \uparrow 0)(e) =e\quad\text{and}\quad
\forall n\in\nat: (Y \uparrow n+1)(e) =Y\big((Y\uparrow n)(e)\big)\;.\]

We refer to~\cite{baaderN98} for the basics of rewriting.
From now on, we fix a finite \emph{signature} $\cF$ together
with an infinite countable set $\cV$ of \emph{variables} with
$\cF\cap\cV=\emptyset$. Elements of $\cF$ (\emph{symbols}) are denoted
by $\f,\g,\h,\zero,\one,\dots$ and elements of $\cV$ by $x,y,z,\dots$
The set of terms over $\cF\cup\cV$ is denoted by $\termset$.
For any $t\in\termset$, we let $\mroot(t)$ denote the
root symbol of $t$: $\mroot(t)=\f$ if $t=\f(t_1,\dots,t_m)$
and $\mroot(t)=\bot$ if $t\in\cV$, where $\bot$ is a new
symbol not occurring in $\cF$ and $\cV$.
We let $\Var(t)$ denote the set of variables
occurring in $t$ and $\pos(t)$ the set of positions of $t$.
For any $p\in\pos(t)$, we write $t|_p$ to denote the subterm of
$t$ at position $p$ and $t[p\la s]$ to denote the term obtained
from $t$ by replacing $t|_p$ with a term $s$. For any $p,q\in\pos(t)$,
we write $p\leq q$ iff $p$ is a prefix of $q$ and we write $p<q$ iff
$p\leq q$ and $p\neq q$. We also define the set of non-variable positions
of $t$ which either are a prefix of $p$ or include $p$ as a prefix:
\begin{linenomath*}
  \[\npos(t,p)=\{q\in\pos(t) \mid
  q\leq p \lor p\leq q,\ t|_q\not\in\cV\}\;.\]
\end{linenomath*}

A \emph{disagreement position} of terms $s$ and $t$ is a position
$p\in\pos(s)\cap\pos(t)$ such that $\mroot(s|_p) \neq \mroot(t|_p)$
and, for every $q<p$, $\mroot(s|_q) = \mroot(t|_q)$. The set of disagreement
positions of $s$ and $t$ is denoted as $\dpos(s,t)$.
A \emph{disagreement pair} of $s$ and $t$ is an ordered pair
$\langle s|_p,t|_p\rangle$ where $p\in\dpos(s,t)$.
\begin{example}
  Let $s=\f(\s(\zero),\s(\one),y)$, $t=\f(x,x,x)$,
  $p_1=1$, $p_2=2$ and $p_3=3$. Then, $\{p_1,p_2\}\subseteq\dpos(s,t)$
  and $\langle s|_{p_1},t|_{p_1}\rangle=\langle \s(\zero),x\rangle$
  and $\langle s|_{p_2},t|_{p_2}\rangle=\langle \s(\one),x\rangle$
  are disagreement pairs of $s$ and $t$. However, $p_3\not\in\dpos(s,t)$
  because $\langle s|_{p_3},t|_{p_3}\rangle=\langle y,x\rangle$
  and $\mroot(y)=\mroot(x)=\bot$.
\end{example}

We write substitutions as sets of the form
$\{x_1/t_1,\dots,x_n/t_n\}$ denoting that for each
$1\leq i \leq n$, variable $x_i$ is mapped to term
$t_i$ (note that $x_i$ may occur in $t_i$). The empty
substitution (identity) is denoted by $\id$. The
application of a substitution $\theta$ to a syntactic
object $o$ is denoted by $o\theta$.
We let $\mgu(s,t)$ denote the (up to variable renaming)
most general unifier of terms $s$ and $t$.
We say that $s$ \emph{semi-unifies} with $t$ when
$s\theta_1\theta_2=t\theta_1$ for some substitutions
$\theta_1$ and $\theta_2$.

A \emph{rewrite rule} (or \emph{rule}) over $\cF\cup\cV$
has the form $l \ra r$ with $l,r\in\termset$,
$l\not\in\cV$ and $\Var(r)\subseteq\Var(l)$.
A \emph{term rewriting system} (TRS) over $\cF\cup\cV$ is a
finite set of rewrite rules over $\cF\cup\cV$.
%
Given a TRS $\cR$ and some terms $s$ and $t$, we write $s\ra_{\cR}t$
if there is a rewrite rule $l\ra r$ in $\cR$, a substitution $\theta$ and
$p \in \pos(s)$ such that $s|_p = l\theta$ and $t=s[p\la r\theta]$.
We let $\ra_{\cR}^+$ (resp. $\ra_{\cR}^*$) denote the transitive (resp.
reflexive and transitive) closure of $\ra_{\cR}$.
We say that a term $t$ is \emph{non-terminating} with respect to (\wrt)
$\cR$ when there exist infinitely many terms $t_1,t_2,\ldots$ such that
$t\ra_{\cR}t_1\ra_{\cR}t_2\ra_{\cR}\cdots$.
We say that $\cR$ is \emph{non-terminating} if there exists
a non-terminating term \wrt{} it. A term $t$ \emph{loops} \wrt{}
$\cR$ when $t\ra_{\cR}^+C[t\theta]$ for some context $C$ and substitution
$\theta$. Then $t\ra_{\cR}^+C[t\theta]$ is called a \emph{loop} for $\cR$.
We say that $\cR$ is \emph{looping} when it admits a loop.
If a term loops \wrt{} $\cR$ then it is non-terminating \wrt{} $\cR$.

The unfolding operators that we define in Sect.~\ref{sect:gunf} of this
paper use \emph{narrowing}. We say that a term $s$ narrows forwards
(resp. backwards) to a term $t$ \wrt{} a TRS $\cR$ when there exists a
non-variable position $p$ of $s$ and a rule $l\ra r$ of $\cR$ renamed
with new variables not previously met such that $t=s[p\leftarrow r]\theta$
(resp. $t=s[p\leftarrow l]\theta$) where $\theta=\mgu(s|_p,l)$
(resp. $\theta=\mgu(s|_p,r)$).

We refer to~\cite{AG00} for details on dependency pairs.
The \emph{defined symbols} of a TRS $\cR$ over $\cF\cup\cV$
are $\cD_{\cR}=\{\mroot(l)\mid l\ra r\in\cR\}$.
For every $\f\in\cF$ we let $\f^{\#}$ be a fresh \emph{tuple symbol}
with the same arity as $\f$. The set of tuple symbols is denoted
as $\cF^{\#}$. The notations and definitions above with terms over
$\cF\cup\cV$ are naturally extended to terms over
$(\cF\cup\cF^{\#})\cup\cV$.
Elements of $\cF\cup\cF^{\#}$ are denoted as $f,g,\dots$
If $t=\f(t_1,\dots,t_m)\in\termset$, we let
$t^{\#}$ denote the term $\f^{\#}(t_1,\dots,t_m)$,
and we call $t^{\#}$ an \emph{$\cF^{\#}$-term}.
An \emph{$\cF^{\#}$-rule} is a rule whose left-hand and
right-hand sides are $\cF^{\#}$-terms.
The set of \emph{dependency pairs} of $\cR$ is
\begin{linenomath*}
  \[\{l^{\#}\ra t^{\#} \mid
  l\ra r\in\cR,\ t\text{ is a subterm of } r,\
  \mroot(t)\in\cD_{\cR}\}\;.\]
\end{linenomath*}
A sequence $s_1\ra t_1,\dots,s_n\ra t_n$
of dependency pairs of $\cR$ is an \emph{$\cR$-chain}
if there exists a substitution $\sigma$ such that
$t_i\sigma\ra^*_{\cR} s_{i+1}\sigma$ holds for every two consecutive
pairs $s_i\ra t_i$ and $s_{i+1}\ra t_{i+1}$ in the sequence.
\begin{theorem}[\cite{AG00}]\label{thm:infinite-chain}
  $\cR$ is non-terminating iff there exists an infinite
  $\cR$-chain.
\end{theorem}

The \emph{dependency graph} of $\cR$ is the graph whose
nodes are the dependency pairs of $\cR$ and there is an arc
from $s\ra t$ to $u\ra v$ iff $s\ra t, u\ra v$
is an $\cR$-chain. This graph is not computable in general
since it is undecidable whether two dependency pairs of $\cR$
form an $\cR$-chain. Hence, for automation, one constructs an
estimated graph containing all the arcs of the real graph.
This is done by computing \emph{connectable terms}, which
form a superset of those terms $s,t$ where
$s\sigma\ra^*_{\cR} t\sigma$ holds for some substitution $\sigma$.
The approximation uses the transformations $\capf$ and $\renf$
where, for any $t\in\termsett$, $\capf(t)$ (resp. $\renf(t)$) results
from replacing all subterms of $t$ with defined root symbol
(resp. all variable occurrences in $t$) by different new variables
not previously met. More formally:
\begin{linenomath*}
\begin{align*}
  \capf(x) &= x \text{ if } x\in\cV\\
  \capf(f(t_1,\dots,t_m)) &= \left\{
  \begin{array}{ll}
    \text{a new variable not previously met} & \text{if } f\in\cD_{\cR}\\
    f(\capf(t_1),\dots,\capf(t_m)) & \text{if } f\not\in\cD_{\cR}
  \end{array}\right.\\
  \renf(x) &= \text{a new variable not previously met}\\
  & \text{\phantom{= } if $x$ is an occurrence of a variable}\\
  \renf(f(t_1,\dots,t_m)) &= f(\renf(t_1),\dots,\renf(t_m))
\end{align*}
\end{linenomath*}
A term $s$ is \emph{connectable} to a term $t$ if
$\renf(\capf(s))$ unifies with $t$. An $\cF^{\#}$-rule
$l\ra r$ is connectable to an $\cF^{\#}$-rule
$s\ra t$ if $r$ is connectable to $s$.
The \emph{estimated dependency graph} of $\cR$ is denoted as
$\dg(\cR)$. Its nodes are the dependency pairs of $\cR$ and
there is an arc from $N$ to $N'$ iff $N$ is connectable to $N'$.
We let $\scc(\cR)$ denote the set of strongly
connected components of $\dg(\cR)$ that contain at least one
arc. Hence, a strongly connected component consisting of a
unique node is in $\scc(\cR)$ only if there is an arc from
the node to itself.

\begin{example}\label{ex:two}
  Let $\cR$ be the TRS of Ex.~\ref{ex:zero}.
  We have $\scc(\cR)=\{\cC\}$ where $\cC$ consists of the node
  $N=\f^{\#}(\s(\zero),\s(\one),x)\ra\f^{\#}(x,x,x)$ and of
  the arc $(N,N)$.
\end{example}

\begin{example}\label{ex:one}
  Let $\cR' = \left\{\f(\zero) \ra \f(\one),\
  \f(\two) \ra \f(\zero),\ \one\ra\zero\right\}$.
  We have $\scc(\cR')=\{\cC'\}$ where $\cC'$ consists of the
  nodes $N_1=\f^{\#}(\zero)\ra\f^{\#}(\one)$ and
  $N_2=\f^{\#}(\two)\ra\f^{\#}(\zero)$ and of the
  arcs $\{N_1,N_2\}\times\{N_1,N_2\}\setminus\{(N_2,N_2)\}$.
  The strongly connected component of $\dg(\cR')$ which consists
  of the unique node $\f^{\#}(\zero)\ra\one^{\#}$
  does not belong to $\scc(\cR')$ because it has no arc.
\end{example}

Finite sequences are written as $[e_1,\dots,e_n]$. We let $\cons$
denote the concatenation operator over finite sequences.
A \emph{path} in $\dg(\cR)$ is a finite sequence
$[N_1,N_2,\dots,N_n]$ of nodes where, for each $1 \leq i < n$,
there is an arc from $N_i$ to $N_{i+1}$. When there is also
an arc from $N_n$ to $N_1$, the path is called a \emph{cycle}.
It is called a \emph{simple cycle} if, moreover, there is no
repetition of nodes (modulo variable renaming).

\section{Guided unfoldings}\label{sect:gunf}
%
In the sequel of this paper, we let $\cR$ denote a TRS over
$\cF\cup\cV$.

While the method sketched in Ex.~\ref{ex:zero-continued}
can be applied directly to the TRS $\cR$ under analysis,
we use a refinement based on the estimated dependency graph
of $\cR$. The cycles in $\dg(\cR)$ are over-approximations of
the infinite $\cR$-chains \ie any infinite $\cR$-chain
corresponds to a cycle in the graph but some cycles in the
graph may not correspond to any $\cR$-chain. Moreover,
by Theorem~\ref{thm:infinite-chain}, if we find an infinite
$\cR$-chain then we have proved that $\cR$ is non-terminating.
Hence, we concentrate on the cycles in $\dg(\cR)$.
We try to \emph{solve} them, \ie to find out if they correspond
to any infinite $\cR$-chain. This is done by iteratively
unfolding the $\cF^{\#}$-rules of the cycles. If the
semi-unification test succeeds on one of the generated unfolded
rules, then we have found a loop.
\begin{definition}[Syntactic loop]
  \label{def:syn-loop}
  A \emph{syntactic loop} in $\cR$ is a finite sequence
  $[N_1,\dots,N_n]$ of distinct (modulo variable renaming)
  $\cF^{\#}$-rules where, for each $1\leq i <n$, $N_i$ is
  connectable to $N_{i+1}$ and $N_n$ is connectable to $N_1$.
  We identify syntactic loops consisting
  of the same (modulo variable renaming) elements, not
  necessarily in the same order.
\end{definition}

Note that the simple cycles in $\dg(\cR)$ are syntactic loops.
For any $\cC \in \scc(\cR)$, we let $\scycles(\cC)$ denote
the set of simple cycles in $\cC$. We also let
\begin{linenomath*}
  \[\scycles(\cR)=\cup_{\cC \in \scc(\cR)}\scycles(\cC)\]
\end{linenomath*}
be the set of simple cycles in $\cR$.
The rules of any simple cycle in $\cR$ are assumed to be
pairwise variable disjoint.

\begin{example}[Ex.~\ref{ex:two} and~\ref{ex:one} continued]
  \label{ex:cycles}
  We have
  \begin{linenomath*}
    \[\scycles(\cR) = \{[N]\} \quad\text{and}\quad
    \scycles(\cR') =\{[N_1],[N_1, N_2]\}\]
  \end{linenomath*}
  with, in $\scycles(\cR')$, $[N_1, N_2] = [N_2, N_1]$.
\end{example}

The operators we use for unfolding an $\cF^{\#}$-rule $R$ at a
disagreement position $p$ are defined as follows.
They are guided by a given term $u$ and they only work on the
non-variable subterms of $R$. They unify a subterm of $R$ with a subterm
of $u$, see~(1) in Def.~\ref{def:forward-unf}-\ref{def:backward-unf}.
This corresponds to what we did in Ex.~\ref{ex:zero-continued} for
generating $V'_1$ from $V_1$, but in the definitions below we do not
only consider $p$, we consider all its prefixes. The operators
also unfold $R$ using narrowing, see (2) in
Def.~\ref{def:forward-unf}-\ref{def:backward-unf}: there,
$l' \ra r' \ll \cR$ means that $l' \ra r'$ is a new occurrence of
a rule of $\cR$ that contains new variables
not previously met. This corresponds to what we did in
Ex.~\ref{ex:zero-continued} for generating $V_1$ from $R_1$.
In contrast to~(1), the positions that are greater than $p$ are also
considered in~(2); for instance in Ex.~\ref{ex:zero-continued}, we
unfolded the inner subterm $\zero$ of the disagreement pair component
$\s(\zero)$.
\begin{definition}[Forward guided unfoldings]
  \label{def:forward-unf}
  Let $l\ra r$ be an $\cF^{\#}$-rule, $s$ be an $\cF^{\#}$-term and
  $p\in\dpos(r,s)$. The forward unfoldings of $l\ra r$ at position $p$,
  guided by $s$ and \wrt{} $\cR$ are
  \begin{linenomath*}
    \begin{align*}
      F_{\cR}(l\ra r, s, p) = & \left\{U \;\middle|\;
      \begin{array}{l}
        q\in\npos(r,p),\ q\leq p\\
        \theta=\mgu(r|_q,s|_q),\ U=(l \ra r)\theta
      \end{array}\right\}^{(1)}\cup\\
      & \left\{U \;\middle|\;
      \begin{array}{l}
        q\in\npos(r,p),\
        l' \ra r' \ll \cR\\
        \theta=\mgu(r|_q,l'),\ U=(l \ra r[q\la r'])\theta
      \end{array}\right\}^{(2)}.
    \end{align*}
  \end{linenomath*}
\end{definition}

\begin{definition}[Backward guided unfoldings]
  \label{def:backward-unf}
  Let $s\ra t$ be an $\cF^{\#}$-rule, $r$ be an $\cF^{\#}$-term and
  $p\in\dpos(r,s)$. The backward unfoldings of $s\ra t$ at position $p$,
  guided by $r$ and \wrt{} $\cR$ are
  \begin{linenomath*}
    \begin{align*}
      B_{\cR}(s\ra t, r, p) = & \left\{U \;\middle|\;
      \begin{array}{l}
        q\in\npos(s,p),\ q\leq p\\
        \theta=\mgu(r|_q,s|_q),\ U=(s \ra t)\theta
      \end{array}\right\}^{(1)}\cup\\
      & \left\{U \;\middle|\;
      \begin{array}{l}
        q\in\npos(s,p),\ l' \ra r' \ll \cR\\
        \theta=\mgu(s|_q,r'),\ U=(s[q\la l'] \ra t)\theta
      \end{array}\right\}^{(2)}.
    \end{align*}
  \end{linenomath*}
\end{definition}

\begin{example}[Ex.~\ref{ex:two} and \ref{ex:cycles} continued]
  \label{ex:unfold}
  $[N]$ is a simple cycle in $\cR$ with
  \begin{linenomath*}
    \[N=\underbrace{\f^{\#}(\s(\zero),\s(\one),x)}_s \ra
    \underbrace{\f^{\#}(x,x,x)}_t\;.\]
  \end{linenomath*}
  Let $r=t$. Then $p=1\in\dpos(r,s)$. Moreover, $q=1.1\in\npos(s,p)$
  because $p\leq q$ and $s|_q=\zero$ is not a variable.
  Let $l'\ra r'=\h\ra\zero\in\cR$. We have $\id=\mgu(s|_q, r')$. Hence, by~(2)
  in Def.~\ref{def:backward-unf}, we have
  \begin{linenomath*}
    \[U_1=\underbrace{\f^{\#}(\s(\h),\s(\one),x)}_{s_1} \ra
    \underbrace{\f^{\#}(x,x,x)}_{t_1}
    \in B_{\cR}(N,r,p)\;.\]
  \end{linenomath*}
  Let $r_1=t_1$. Then, $p=1\in\dpos(r_1,s_1)$. Moreover, $p\in\npos(s_1,p)$
  with $s_1|_p=\s(\h)$, $p\leq p$ and $r_1|_p=x$.
  As $\{x/\s(\h)\}=\mgu(r_1|_p,s_1|_p)$, by~(1) in Def.~\ref{def:backward-unf}
  we have
  \begin{linenomath*}
    \[U'_1=\underbrace{\f^{\#}(\s(\h),\s(\one),\s(\h))}_{s'_1} \ra
    \underbrace{\f^{\#}(\s(\h),\s(\h),\s(\h))}_{t'_1}
    \in B_{\cR}(U_1,r_1,p)\;.\]
  \end{linenomath*}
  Let $r'_1=t'_1$. Then, $p'=2.1\in\dpos(r'_1,s'_1)$ with
  $p'\in\npos(s'_1,p')$. Let $l{''}\ra r{''}=\h\ra\one\in\cR$.
  We have $\id=\mgu(s'_1|_{p'}, r{''})$. Hence, by~(2)
  in Def.~\ref{def:backward-unf}, we have
  \begin{linenomath*}
    \[U^{''}_1=\f^{\#}(\s(\h),\s(\h),\s(\h)) \ra
    \f^{\#}(\s(\h),\s(\h),\s(\h)) \in B_{\cR}(U'_1,r'_1,p')\;.\]
  \end{linenomath*}
\end{example}

Our approach consists of iteratively unfolding syntactic
loops using the following operator.
\begin{definition}[Guided unfoldings]
  \label{def:gunf-op}
  Let $X$ be a set of syntactic loops of $\cR$.
  The \emph{guided unfoldings} of $X$ \wrt{} $\cR$ are
  defined as
  \begin{linenomath*}
  \begin{align*}
    \gunfop_{\cR}(X) =
    & \left\{L\cons[U]\cons L' \;\middle|\;
    \begin{array}{l}
      L\cons[l\ra r,s\ra t]\cons L' \in X,\
      \theta = \mgu(r,s)\\
      U = (l \ra t)\theta,\
      L\cons[U]\cons L'\text{ is a syntactic loop}
    \end{array}
    \right\}^{(1)} \cup\\
    & \left\{L\cons[U,s\ra t]\cons L' \;\middle|\;
    \begin{array}{l}
      L\cons[l\ra r,s\ra t]\cons L' \in X\\
      p \in \dpos(r,s),\
      U  \in F_{\cR}(l\ra r, s, p)\\{}
      L\cons[U,s\ra t]\cons L' \text{ is a syntactic loop}
    \end{array}
    \right\}^{(2)} \cup\\
    & \left\{L\cons[l\ra r,U]\cons L' \;\middle|\;
    \begin{array}{l}
      L\cons[l\ra r,s\ra t]\cons L' \in X\\
      p \in \dpos(r,s),\
      U \in B_{\cR}(s\ra t, r, p)\\{}
      L\cons[l\ra r,U]\cons L' \text{ is a syntactic loop}
    \end{array}
    \right\}^{(3)} \cup\\
    & \left\{[U] \;\middle|\;
    \begin{array}{l}
      [l\ra r]\in X,\
      p \in \dpos(r,l)\\
      U \in F_{\cR}(l\ra r, l, p) \cup B_{\cR}(l\ra r, r, p)\\{}
      [U] \text{ is a syntactic loop}
    \end{array}
    \right\}^{(4)}.
  \end{align*}
  \end{linenomath*}
\end{definition}
The general idea is to \emph{compress} the syntactic loops into singletons
by iterated applications of this operator.
The semi-unification criterion can then be applied to these singletons,
see Theorem~\ref{theorem:nonterm-criterion} below. Compression takes place in
case~(1) of Def.~\ref{def:gunf-op}: when the right-hand side of a rule unifies
with the left-hand side of its successor, then both rules are merged. When
merging two successive rules is not possible yet, the operators $F_{\cR}$ and
$B_{\cR}$ are applied to try to transform the rules into mergeable ones,
see cases~(2) and~(3). Once a syntactic loop has been compressed to a
singleton, we keep on unfolding (case~(4)) to try reaching a compressed
form satisfying the semi-unification criterion.
Note that after an unfolding step, we might get a sequence which is not a
syntactic loop: the newly generated rule $U$ might be identical to another
rule in the sequence or it might not be connectable to its predecessor or
successor. So, (1)--(4) require that the generated sequence is a syntactic loop.

The guided unfolding semantics is defined as follows, in the
style of~\cite{alpuenteFMV97,payet08}.
\begin{definition}[Guided unfolding semantics]
  \label{def:gunf}
  The \emph{guided unfolding semantics} of $\cR$ is the limit
  of the unfolding process described in Def.~\ref{def:gunf-op},
  starting from the simple cycles in $\cR$:
  $\gunf(\cR) = \cup_{n\in\nat}\gunf(\cR,n)$ where, for all
  $n\in\nat$,
  \begin{linenomath*}
    \[\gunf(\cR,n) = (\gunfop_{\cR}\uparrow n)(\scycles(\cR))\;.\]
  \end{linenomath*}
\end{definition}
This semantics is very similar to the \emph{overlap closure}~\cite{Guttag83}
of $\cR$ (denoted by $OC(\cR)$). A difference is that for computing
$\gunf(\cR)$ one starts from dependency pairs of $\cR$ ($\scycles(\cR)$),
whereas for computing $OC(\cR)$ one starts directly from the rules of $\cR$.
In case~(1) of Def.~\ref{def:gunf-op}, we merge two unfolded rules.
Similarly, for computing $OC(\cR)$ one overlaps closures with closures.
However, in cases~(2)--(4) the operators $F_{\cR}$ and $B_{\cR}$ narrow an
unfolded rule with a rule of $\cR$, not with another unfolded rule, unlike
in the computation of $OC(\cR)$.

\begin{example}\label{ex:gunf1}
  By Ex.~\ref{ex:unfold} and $(4)$ in Def.~\ref{def:gunf-op},
  we have $[U^{''}_1]\in \gunf(\cR,3)$.
\end{example}

\begin{example}\label{ex:gunf2}
  Let $\cR = \{\f(\zero) \ra \g(\one),\ \g(\one) \ra \f(\zero)\}$.
  Then, $\scc(\cR)=\{\cC\}$ where $\cC$ consists of the nodes
  $N_1=\f^{\#}(\zero)\ra\g^{\#}(\one)$ and
  $N_2=\g^{\#}(\one)\ra\f^{\#}(\zero)$ and of the arcs
  $(N_1,N_2)$ and $(N_2,N_1)$. Moreover,
  $\scycles(\cR) =\{[N_1, N_2]\}$.
  As $\id=\mgu(\g^{\#}(\one),\g^{\#}(\one))$ and
  $(\f^{\#}(\zero)\ra\f^{\#}(\zero))\id =
  \f^{\#}(\zero)\ra\f^{\#}(\zero)$,
  by~$(1)$ in Def.~\ref{def:gunf-op} we have
  $[\f^{\#}(\zero)\ra\f^{\#}(\zero)]\in\gunf(\cR,1)$.
\end{example}

\begin{proposition}
  \label{prop:rewritesto}
  For any $n\in\nat$ and $[s^{\#}\ra t^{\#}]\in\gunf(\cR,n)$
  there exists some context $C$ such that $s \ra_{\cR}^+ C[t]$.
\end{proposition}
\begin{proof}
  For some context $C$, we have $s \ra C[t]\in\mathit{unf}(\cR)$
  where $\mathit{unf}(\cR)$ is the unfolding semantics defined
  in~\cite{payet08}. So, by Prop.~3.12 of~\cite{payet08},
  $s \ra_{\cR}^+ C[t]$.
\end{proof}

\section{Inferring terms that loop}\label{sect:loops}
As in~\cite{payet08}, we use semi-unification~\cite{lankford78}
for detecting loops. Semi-unification encompasses both matching
and unification, and a polynomial-time algorithm for it can be
found in~\cite{Kapur91}.
\begin{theorem}\label{theorem:nonterm-criterion}
  For any $n\in\nat$, if there exist
  $[s^{\#}\ra t^{\#}] \in \gunf(\cR,n)$
  and some substitutions $\theta_1$ and $\theta_2$
  such that $s\theta_1\theta_2 = t\theta_1$, then the term
  $s\theta_1$ loops \wrt{} $\cR$.
\end{theorem}
\begin{proof}
  By Prop.~\ref{prop:rewritesto},
  $s \ra_{\cR}^+ C[t]$ for some context $C$.
  Since $\ra_{\cR}$ is stable, we have
  \begin{linenomath*}
    \[s\theta_1 \ra^+_{\cR} C[t]\theta_1 \quad\ie\quad
    s\theta_1 \ra^+_{\cR} C\theta_1[t\theta_1] \quad\ie\quad
    s\theta_1 \ra^+_{\cR} C\theta_1[s\theta_1\theta_2]\;.\]
  \end{linenomath*}
  Hence, $s\theta_1$ loops \wrt{} $\cR$.
\end{proof}

\begin{example}[Ex.~\ref{ex:gunf1} continued]
  \label{ex:nonterm1}
  We have
  \begin{linenomath*}
    \[[\underbrace{\f^{\#}(\s(\h),\s(\h),\s(\h))\ra\f^{\#}(\s(\h),\s(\h),\s(\h))}_{U^{''}_1}]
    \in \gunf(\cR,3)\]
  \end{linenomath*}
  with $\f(\s(\h),\s(\h),\s(\h))\theta_1\theta_2=
  \f(\s(\h),\s(\h),\s(\h))\theta_1$
  for $\theta_1=\theta_2=\id$. Consequently,
  $\f(\s(\h),\s(\h),\s(\h))\theta_1 =
  \f(\s(\h),\s(\h),\s(\h))$ loops
  \wrt{} $\cR$.
\end{example}

\begin{example}[Ex.~\ref{ex:gunf2} continued]
  \label{ex:nonterm2}
  $[\f^{\#}(\zero)\ra\f^{\#}(\zero)] \in \gunf(\cR,1)$ with
  $\f(\zero)\theta_1\theta_2=\f(\zero)\theta_1$
  for $\theta_1=\theta_2=\id$. Hence,
  $\f(\zero)\theta_1 = \f(\zero)$ loops \wrt{} $\cR$.
\end{example}

The substitutions $\theta_1$ and $\theta_2$ that
we use in the next example are more sophisticated than
in Ex.~\ref{ex:nonterm1} and Ex.~\ref{ex:nonterm2}.
\begin{example}\label{ex:semi-unification}
  Let $\cR = \{\f(\g(x,\zero),y) \ra \f(\g(\zero,x),\h(y))\}$.
  Then, $\scc(\cR)=\{\cC\}$ where $\cC$ consists of the node
  $N=\f^{\#}(\g(x,\zero),y)\ra\f^{\#}(\g(\zero,x),\h(y))$ and
  of the arc $(N,N)$. Moreover, $\scycles(\cR) =\{[N]\}$ hence
  $[N]\in\gunf(\cR,0)$. Therefore, as
  $\f(\g(x,\zero),y)\theta_1\theta_2=\f(\g(\zero,x),\h(y))\theta_1$
  for $\theta_1=\{x/\zero\}$ and $\theta_2=\{y/\h(y)\}$,
  by Theorem~\ref{theorem:nonterm-criterion} we have that
  $\f(\g(x,\zero),y)\theta_1 = \f(\g(\zero,\zero),y)$ loops
  \wrt{} $\cR$.
\end{example}

We do not have an example where semi-unification is necessary for detecting
a loop. In every example that we have considered, matching or unification
were enough. However, semi-unification sometimes allows us to detect loops
earlier in the unfolding process than with matching and unification. This is
important in practice, because the number of unfolded rules can grow rapidly
from iteration to iteration.
\begin{example}[Ex.~\ref{ex:semi-unification} continued]
  Semi-unification allows us to detect a loop at iteration 0 of $\gunfop_{\cR}$.
  But a loop can also be detected at iteration 1 using matching. Indeed, we have
  \begin{linenomath*}
    \[N=\underbrace{\f^{\#}(\g(x,\zero),y)}_l\ra
    \underbrace{\f^{\#}(\g(\zero,x),\h(y))}_r\]
  \end{linenomath*}
  with $p=1.1\in\dpos(r,l)$. Hence,
  \begin{linenomath*}
    \[U=\underbrace{\f^{\#}(\g(\zero,\zero),y)}_{s^{\#}}\ra
    \underbrace{\f^{\#}(\g(\zero,\zero),\h(y))}_{t^{\#}}
    \in F_{\cR}(l\ra r, l, p)\;.\]
  \end{linenomath*}
  So, by~$(4)$ in Def.~\ref{def:gunf-op}, we have
  $[U]\in \gunf(\cR,1)$. Notice that
  $s\theta=t$ for $\theta=\{y/\h(y)\}$, so $s$ matches $t$.
  Moreover, by Theorem~\ref{theorem:nonterm-criterion},
  $s=\f(\g(\zero,\zero),y)$ loops \wrt{} $\cR$ (take $\theta_1=\id$
  and $\theta_2=\theta$).
\end{example}

\section{Further comparisons with the approach of~\cite{payet08}}
\label{sect:comparison08}
The approach that we have presented in~\cite{payet08} relies on an unfolding
operator $U_{\cR}$ (where $\cR$ is the TRS under analysis) which is also based
on forward and backward narrowing (as $F_{\cR}$ and $B_{\cR}$ herein).
But, unlike the technique that we have presented above, it directly unfolds
the rules (not the dependency pairs) of $\cR$ and it does not compute any SCC.
Moreover, it consists of a thorough computation of the iterations of $U_{\cR}$
followed by a mechanism for eliminating rules that cannot be further unfolded
to a rule $l\ra r$ where $l$ semi-unifies with $r$. Such rules are said to be
\emph{root-useless}. The set of root-useless rules is an overapproximation of
the set of useless rules (rules that cannot contribute to detecting a loop),
hence the elimination technique of~\cite{payet08} may remove some rules which
are actually useful for detecting a loop. Our non-termination analyser which
is based on~\cite{payet08} uses a time limit. It stops whenever it has detected
a loop within the limit (then it answers NO, standing for
\emph{No, this TRS does not terminate}
as in the Termination Competition~\cite{termcomp}) or when the limit
has been reached (then it answers TIME OUT) or when no more unfolded
rule could be generated at some point within the limit
($(U_{\cR}\uparrow n)(\cR)=\emptyset$ for some $n$).
In the last situation, either the TRS under analysis is not looping
(it is terminating or non-looping non-terminating)
or it is looping but a loop for it cannot be captured by the approach
(for instance, the elimination mechanism has removed all the useful rules).
In such a situation, our analyser answers DON'T KNOW.
\begin{example}\label{ex:dontknow}
  Consider the terminating TRS $\cR=\{\zero\ra\one\}$.
  As the left-hand (resp. right-hand) side of the rule of $\cR$
  cannot be narrowed backwards (resp. forwards) with $\cR$ then we have
  $(U_{\cR}\uparrow 1)(\cR)=\emptyset$.
\end{example}

In contrast, the approach that we have presented in
Sect.~\ref{sect:gunf}--\ref{sect:loops} above avoids the generation of
some rules by only unfolding disagreement pairs.
Currently, in terms of loop detection power, we do not have any theoretical
comparison between this new technique and that of~\cite{payet08}.
Our new non-termination analyser also uses a time limit and answers NO,
TIME OUT or DON'T KNOW when no more unfolded rules are generated at some
point ($\gunf(\cR,n)=\emptyset$ for some $n$, as in Ex.~\ref{ex:dontknow}).
Moreover, it allows the user to fix a \emph{selection strategy} of disagreement
pairs: in Def.~\ref{def:gunf-op}, the conditions $p\in\dpos(r,s)$
(cases~(2)--(3)) and $p\in\dpos(r,l)$ (case~(4)) are replaced
with $p\in\select_{\cR}(l\ra r,s\ra t)$ and $p\in\select_{\cR}(l\ra r,l\ra r)$
respectively, where $\select_{\cR}$ can be one of the following functions.
\begin{description}
  \item[Selection of all the pairs:] $\selecta_{\cR}(l\ra r,s\ra t)=\dpos(r,s)$.
  \item[Leftmost selection:] if $\dpos(r,s)=\emptyset$ then
  $\selectlm_{\cR}(l\ra r,s\ra t)=\emptyset$, otherwise
  $\selectlm_{\cR}(l\ra r,s\ra t)=\{p\}$ where $p$ is the leftmost disagreement
  position of $r$ and $s$.
  \item [Leftmost selection with non-empty unfoldings:]
  \begin{linenomath*}
    \[\selectlmne_{\cR}(l\ra r,s\ra t)=\{p\}\]
  \end{linenomath*}
  where $p$ is the leftmost disagreement position of $r$ and $s$ such that
  \begin{linenomath*}
    \[F_{\cR}(l\ra r, s, p)\cup B_{\cR}(s\ra t, r, p)\neq\emptyset\;.\]
  \end{linenomath*}
  If such a position $p$ does not exist then
  $\selectlmne_{\cR}(l\ra r,s\ra t)=\emptyset$.
\end{description}
\begin{example}
  Let $\cR=\{\f(\s(\zero),\s(\one),z)\ra \f(x,y,z)\}$ and
  $l=\f^{\#}(\s(\zero),\s(\one),z)$ and $r=\f^{\#}(x,y,z)$.
  Then, we have
  $\selecta_{\cR}(l\ra r,l\ra r)=\dpos(r,l)=\{1,2\}$.
  Moreover, $1$ is the leftmost disagreement position of $r$ and $l$ because
  $r|_1=x$ occurs to the left of $r|_2=y$ in $r$ and $l|_1=\s(\zero)$ occurs
  to the left of $l|_2=\s(\one)$ in $l$. Therefore, we have
  $\selectlm_{\cR}(l\ra r,l\ra r) = \{1\}$.
\end{example}
\begin{example}
  \label{ex:lmne}
  Let $\cR = \left\lbrace \f(x,x) \ra \f(\g(x),\h(x)),\ \h(x) \ra \g(x)
  \right\rbrace$. Then, $\scc(\cR)=\{\cC\}$ where $\cC$ consists of the node
  $N=l\ra r=\f^{\#}(x,x)\ra\f^{\#}(\g(x),\h(x))$ and of the arc $(N,N)$.
  Then, $\dpos(r,l)=\{1,2\}$ and $\selectlm_{\cR}(N,N) = \{1\}$.
  As $F_{\cR}(N,l,1)\cup B_{\cR}(N,r,1)=\emptyset$ and
  $F_{\cR}(N,l,2)\cup B_{\cR}(N,r,2)\neq\emptyset$
  (for instance, $\f^{\#}(x,x)\ra\f^{\#}(\g(x),\g(x))\in F_{\cR}(N,l,2)$
  is obtained from narrowing $r|_2=\h(x)$ forwards with $\h(x) \ra \g(x)$),
  then $\selectlmne_{\cR}(N,N) = \{2\}$.
\end{example}

As the approach of~\cite{payet08}, and depending on the strategy used for
selecting disagreement pairs, our new technique is not complete in the sense
that it may miss some loop witnesses.
\begin{example}[Ex.~\ref{ex:lmne} continued]
  \label{ex:uncomplete}
  We have $\gunf(\cR,0)=\scycles(\cR)=\{[N]\}$.
  As $l$ does not semi-unify with $r$, no loop is detected
  from $\gunf(\cR,0)$, so we go on and compute $\gunf(\cR,1)$.
  Only case~(4) of Def.~\ref{def:gunf-op} is applicable to $[N]$.
  First, suppose that $\select_{\cR}=\selectlm_{\cR}$. Then,
  $\select_{\cR}(N,N)=\{1\}$ and, as
  $F_{\cR}(N,l,1)\cup B_{\cR}(N,r,1)=\emptyset$,
  case~(4) does not produce any rule.
  Consequently, we have $\gunf(\cR,1) = \emptyset$, hence no loop is detected
  for $\cR$. Now, suppose that $\select_{\cR}=\selectlmne_{\cR}$.
  Then, $\select_{\cR}(N,N)=\{2\}$.
  Narrowing $r|_2=\h(x)$ forwards with $\h(x) \ra \g(x)$, we get the
  rule $N'=\f^{\#}(x,x)\ra\f^{\#}(\g(x),\g(x))$ which is an element of
  $F_{\cR}(N,l,2)$. Hence, $[N']\in\gunf(\cR,1)$.
  As in $N'$ we have that $\f(x,x)$ semi-unifies with $\f(\g(x),\g(x))$
  (take $\theta_1=\id$ and $\theta_2=\{x/\g(x)\}$), then $\f(x,x)$
  loops \wrt{} $\cR$.
  This loop is also detected by the approach of~\cite{payet08}.
\end{example}

\section{Experiments}\label{sect:experiments}
We have implemented the technique of this paper in our analyser \nti
\footnote{
{\tt\url{http://lim.univ-reunion.fr/staff/epayet/Research/NTI/NTI.html}}}
(Non-Termination Inference).
For our experiments, we have extracted from the directory
\verb+TRS_Standard+ of the TPBD~\cite{TPDB} all the valid rewrite
systems\footnote{Surprisingly, the subdirectory
\texttt{Transformed\_CSR\_04} contains 60~files where a pair
$l\ra r$ with $\Var(r)\not\subseteq\Var(l)$ occurs. These pairs
are not valid rewrite rules.}
that were either proved looping or unproved\footnote{By \emph{unproved}
we mean that no tool succeeded in proving that these TRSs were terminating
or non-terminating (all the tools failed on these TRSs).}
during the Termination Competition~2017 (TC'17)~\cite{termcomp}.
Otherwise stated, we removed from \verb+TRS_Standard+
all the non-valid TRSs and all the TRSs that were proved terminating or
non-looping non-terminating by a tool participating in the competition.
We ended up with a set $\cS$ of 333~rewrite systems. We let $\cL$
(resp. $\cU$) be the subset of $\cS$ consisting of all the systems that were
proved looping (resp. that were unproved) during TC'17. Some characteristics
of $\cL$ and $\cU$ are reported in Table~\ref{tab:characteristics}.
Note that the complete set of simple cycles of a TRS may be really huge,
hence \nti{} only computes a subset of it. The simple cycle characteristics
given in Table~\ref{tab:characteristics} relate to the subsets
computed by \nti.

\begin{table}
  \begin{center}
    \begin{tabular}{|c||rr|rr|r||rr|rr|r|}
      \hline
      \multirow{2}{*}{$\cS=\cL\uplus\cU$ \scriptsize{(333~TRSs)}}
      & \multicolumn{5}{c||}{$\cL$ \scriptsize{(173~TRSs)}}
      & \multicolumn{5}{c|}{$\cU$ \scriptsize{(160~TRSs)}}\\
      \cline{2-11}
      & \multicolumn{2}{c|}{\textbf{Min}}
      & \multicolumn{2}{c|}{\textbf{Max}}
      & \multicolumn{1}{c||}{\textbf{Average}}
      & \multicolumn{2}{c|}{\textbf{Min}}
      & \multicolumn{2}{c|}{\textbf{Max}}
      & \multicolumn{1}{c|}{\textbf{Average}} \\
      \hline\hline
      \textbf{TRS size} & 1 & {\scriptsize{}[17]} & 104 & {\scriptsize{}[1]} & 11.08
      & 1 & {\scriptsize{}[9]} & 837 & {\scriptsize{}[1]} & 64.78 \\
      \hline\hline
      \textbf{Number of SCCs} & 1 & {\scriptsize{}[101]} & 12 & {\scriptsize{}[1]} & 1.95
      & 1 & {\scriptsize{}[70]} & 130 & {\scriptsize{}[1]} & 6.14 \\
      \hline
      \textbf{SCC size} & 1 & {\scriptsize{}[96]} & 192 & {\scriptsize{}[1]} & 4.44
      & 1 & {\scriptsize{}[54]} & 473 & {\scriptsize{}[1]} & 10.79 \\
      \hline\hline
      \textbf{Number of simple cycles} & 1 & {\scriptsize{}[47]} & 185 & {\scriptsize{}[1]} & 8.55
      & 2 & {\scriptsize{}[15]} & 1,176 & {\scriptsize{}[1]} & 69.02 \\
      \hline
      \textbf{Simple cycle size} & 1 & {\scriptsize{}[157]} & 9 & {\scriptsize{}[2]} & 2.23
      & 1 & {\scriptsize{}[157]} & 9 & {\scriptsize{}[4]} & 2.21 \\
      \hline\hline
      \textbf{Number of symbols} & 1 & {\scriptsize{}[4]} & 66 & {\scriptsize{}[1]} & 9.09
      & 2 & {\scriptsize{}[7]} & 259 & {\scriptsize{}[1]} & 24.14 \\
      \hline
      \textbf{Symbol arity} & 0 & {\scriptsize{}[153]} & 5 & {\scriptsize{}[2]} & 1.07
      & 0 & {\scriptsize{}[150]} & 12 & {\scriptsize{}[2]} & 1.94 \\
      \hline\hline
      \textbf{Number of defined symbols}
      & 1 & {\scriptsize{}[28]} & 58 & {\scriptsize{}[1]} & 5.21
      & 1 & {\scriptsize{}[15]} & 132 & {\scriptsize{}[2]} & 17.19 \\
      \hline
      \textbf{Defined symbol arity} & 0 & {\scriptsize{}[74]} & 5 & {\scriptsize{}[2]} & 1.38
      & 0 & {\scriptsize{}[28]} & 12 & {\scriptsize{}[2]} & 2.27 \\
      \hline
    \end{tabular}
  \end{center}
  \caption{Some characteristics of the analysed TRSs.
  Sizes are in number of rules. In square brackets, we report the
  number of TRSs with the corresponding min or max.}
  \label{tab:characteristics}
\end{table}

We have run our new approach (\nti'18) and that of~\cite{payet08}
(\nti'08) on the TRSs of $\cS$. The results are reported in
Table~\ref{tab:experiments-L} and Table~\ref{tab:experiments-U}.
We used an Intel 2-core i5 at 2~GHz with 8~GB of RAM and
the time limit fixed for a proof was 120s.
For every selection strategy, \nti'18 issues more successful proofs (NO)
and generates less unfolded rules than \nti'08. Moreover, as it avoids the
generation of some rules instead of computing all the unfolding and
then eliminating some rules (as \nti'08 does), its times are better.
%
At the bottom of the tables, we give the numbers of TRSs proved looping by
both approaches and by one approach only.
\nti'18 succeeds on all the TRSs of $\cL$ on which \nti'08 succeeds,
but it fails on one TRS of $\cU$ on which \nti'08 succeeds.
This is due to our simplified computation of the set of simple cycles:
our algorithm does not generate the cycle that would allow \nti'18 to succeed
and \nti'18 times out, trying to unfold syntactic loops from which it
cannot detect anything.
Another point to note is that the implementation of the new approach
does not need to run several analyses in parallel to achieve the results
presented in Table~\ref{tab:experiments-L} and Table~\ref{tab:experiments-U}.
One single thread of computation is enough. On the contrary, for the approach
of~\cite{payet08}, 3 parallel threads are necessary: one with forward
unfoldings only, one with backward unfoldings only and one with forward and
backward unfoldings together.
The results get worse if \nti'08 only runs one thread performing forward
and backward unfoldings together.
In Table~\ref{tab:experiments-L} and Table~\ref{tab:experiments-U},
we report in square brackets the number of successes of \nti'08 when
it only runs one thread performing both forward and backward unfoldings.

\begin{table}[ht]
  \begin{center}
    \begin{tabular}{|c||c||r|r|r|}
      \hline
      \multirow{2}{*}{$\cL$ \scriptsize{(173~TRSs)}}
      & \multirow{2}{*}{\nti'08}
      & \multicolumn{3}{|c|}{\nti'18}\\
      \cline{3-5}
      & & $\selecta$ & $\selectlm$ & $\selectlmne$ \\
      \hline\hline
      \textbf{NO} & \multicolumn{1}{|r||}{152 {\scriptsize{}[149]}}
      & 157 & 157 & 158\\
      \hline
      \textbf{DON'T KNOW} & \multicolumn{1}{|r||}{0}
      & 0 & 1 & 0 \\
      \hline
      \textbf{TIME OUT} & \multicolumn{1}{|r||}{21}
      & 16 & 15 & 15 \\
      \hline\hline
      \textbf{Time} & \multicolumn{1}{|r||}{2,966s}
      & 2,194s & 1,890s & 1,889s \\
      \hline
      \textbf{Generated rules} & \multicolumn{1}{|r||}{11,167,976}
      & 9,030,962 & 8,857,421 & 8,860,560 \\
      \hline\hline
      \multicolumn{2}{|c||}{$\text{{\bf{}NO}(\nti'08)} \cap \text{{\bf{}NO}(\nti'18)}$}
      & 152 & 151 & 152 \\
      \hline
      \multicolumn{2}{|c||}{$\text{{\bf{}NO}(\nti'08)} \setminus \text{{\bf{}NO}(\nti'18)}$}
      & 0 & 1 & 0\\
      \hline
      \multicolumn{2}{|c||}{$\text{{\bf{}NO}(\nti'18)} \setminus \text{{\bf{}NO}(\nti'08)}$}
      & 5 & 6 & 6\\
      \hline
    \end{tabular}
  \end{center}
  \caption{Analysis results on the TRSs of $\cL$.}
  \label{tab:experiments-L}
\end{table}

\begin{table}[ht]
  \begin{center}
    \begin{tabular}{|c||c||r|r|r|}
      \hline
      \multirow{2}{*}{$\cU$ \scriptsize{(160~TRSs)}}
      & \multirow{2}{*}{\nti'08}
      & \multicolumn{3}{|c|}{\nti'18}\\
      \cline{3-5}
      & & $\selecta$ & $\selectlm$ & $\selectlmne$ \\
      \hline\hline
      \textbf{NO} & \multicolumn{1}{|r||}{4 {\scriptsize{}[3]}}
      & 6 & 6 & 6 \\
      \hline
      \textbf{DON'T KNOW} & \multicolumn{1}{|r||}{0}
      & 0 & 1 & 0 \\
      \hline
      \textbf{TIME OUT} & \multicolumn{1}{|r||}{156}
      & 154 & 153 & 154 \\
      \hline\hline
      \textbf{Time} & \multicolumn{1}{|r||}{18,742s}
      & 18,563s & 18,414s & 18,534s \\
      \hline
      \textbf{Generated rules} & \multicolumn{1}{|r||}{64,011,002}
      & 53,134,334 & 61,245,705 & 63,300,604 \\
      \hline\hline
      \multicolumn{2}{|c||}{$\text{{\bf{}NO}(\nti'08)} \cap \text{{\bf{}NO}(\nti'18)}$}
      & 3 & 3 & 3 \\
      \hline
      \multicolumn{2}{|c||}{$\text{{\bf{}NO}(\nti'08)} \setminus \text{{\bf{}NO}(\nti'18)}$}
      & 1 & 1 & 1\\
      \hline
      \multicolumn{2}{|c||}{$\text{{\bf{}NO}(\nti'18)} \setminus \text{{\bf{}NO}(\nti'08)}$}
      & 3 & 3 & 3\\
      \hline
    \end{tabular}
  \end{center}
  \caption{Analysis results on the TRSs of $\cU$.}
  \label{tab:experiments-U}
\end{table}

Four tools participated in the category \emph{TRS Standard}
of TC'17: \aprove{}~\cite{aproveWeb,aprove}, \muterm{}~\cite{mutermWeb},
\natt{}~\cite{nattWeb} and \wanda{}~\cite{wandaWeb}.
The numbers of TRSs proved looping by each of them during the competition
is reported in Table~\ref{tab:competition}. An important point to note here
is that the time limit fixed in TC'17 was 300s, whereas in our experiments with
\nti'18 and \nti'08 it was 120s. Moreover, the machine we used
(an Intel 2-core i5 at 2~GHz with 8~GB of RAM) is much less powerful than
the machine used during TC'17 (the StarExec platform~\cite{starexec}
running on an Intel Xeon E5-2609 at 2.4~GHz with 129~GB of RAM).
All the tools of TC'17 failed on all the rewrite systems of $\cU$. In contrast,
\nti'18 (resp. \nti'08) finds a loop for 6 (resp. 4) of them.
Regarding $\cL$, \aprove{}~was able to prove loopingness of 172 out of 173~TRSs.
The only TRS of $\cL$ on which \aprove{} failed\footnote{\texttt{Ex6\_15\_AEL02\_FR.xml}
in the directory \texttt{TRS\_Standard/Transformed\_CSR\_04}} was proved looping
by \natt{}. In comparison, our approach succeeds on 158~systems of $\cL$,
less than \aprove{} but more than the other tools of TC'17.
Similarly to our approach, \aprove{} handles the SCCs of the estimated
dependency graph independently, but it first performs a termination analysis.
The non-termination analysis is then only applied to those SCCs that could
not be proved terminating.
On the contrary, \nti{} only performs non-termination proofs. If an SCC
is terminating, it cannot prove it and keeps on trying a non-termination
proof, unnecessarily generating unfolded rules at the expense of the
analysis of the other SCCs.
The loop detection techniques implemented in \aprove{} and \nti{}'18
are based on the idea of searching for loops by forward and backward narrowing
of dependency pairs and by using semi-unification to detect potential loops.
This idea has been presented in~\cite{gieslTS05} where heuristics are used to
select forward or backward narrowing. Note that in constrast, the technique
that we present herein does not use any heuristics and proceeds both forwards
and backwards.

\begin{table}
  \begin{center}
    \begin{tabular}{|c||r|r|r|r||r|r|}
      \hline
      \multirow{2}{*}{$\cS=\cL\uplus\cU$ \scriptsize{(333~TRSs)}}
      & \multicolumn{4}{c||}{TC'17 \scriptsize{(time limit = 300s)}}
      & \multicolumn{2}{c|}{\scriptsize{(time limit = 120s)}} \\
      \cline{2-7}
      & \multicolumn{1}{c|}{\aprove} & \multicolumn{1}{c|}{\muterm}
      & \multicolumn{1}{c|}{\natt} & \multicolumn{1}{c||}{\wanda}
      & \multicolumn{1}{c|}{\nti'08} &
      \multicolumn{1}{c|}{\nti'18 \scriptsize{$(\selectlmne)$}} \\
      \hline\hline
      $\cL$ \scriptsize{(173~TRSs)} & 172 & 81 & 109 & 0 & 152 & 158 \\
      \hline
      $\cU$ \scriptsize{(160~TRSs)} & 0 & 0 & 0 & 0 & 4 & 6 \\
      \hline
    \end{tabular}
  \end{center}
\caption{Number of successes (NO) on $\cL$ and $\cU$ obtained during
TC'17 and those obtained by \nti{} during our experiments.}
\label{tab:competition}
\end{table}

\section{Conclusion}
We have reconsidered and modified the unfolding-based technique
of~\cite{payet08} for detecting loops in standard term rewriting. The new
approach uses disagreement pairs for guiding the unfoldings, which now are
only partially computed, whereas the technique
of~\cite{payet08} consists of a thorough computation followed by a mechanism
for eliminating some rules. Two theoretical questions remain open: in terms
of loop detection, is an approach more powerful than the other and does
semi-unification subsume matching and unification?

We have implemented the new approach in our tool \nti{} and compared it
to~\cite{payet08} on a set of 333~rewrite systems. The new results are
better (better times, more successful proofs, less unfolded rules).
Moreover, the approach compares well to the tools that participated in TC'17.
However, the number of generated rules is still important. In an attempt
to reduce it, during our experiments we added the elimination mechanism
of~\cite{payet08} to the new approach, but the results we recorded were
not satisfactory (an equivalent, slightly smaller, number of generated
rules but, due to the computational overhead, bigger times and less successes);
hence, we removed it.
Termination analysis may help to reduce the number of unfolded rules
by detecting terminating SCCs in the estimated dependency graph \ie
SCCs on which it is useless to try a non-termination proof.
In other words, we could use termination analysis
as an elimination mechanism. Several efficient and
powerful termination analysers have been implemented so far~\cite{tools}
and one of them could be called by \nti.
A final idea to improve our approach would be to consider more
sophisticated strategies for selecting disagreement pairs.

\section*{Acknowledgements}
The author thanks the anonymous reviewers for their many helpful
comments and constructive criticisms. He also thanks Fred~Mesnard
for presenting the paper at the symposium.

\bibliographystyle{plain}

\end{document}